\documentclass[11pt, a4paper,reqno]{amsart}

\usepackage{color} \definecolor{bleu_sombre}{rgb}{0,0,0.6}  \definecolor{rouge_sombre}{rgb}{0.8,0,0}\definecolor{vert_sombre}{rgb}{0,0.6,0}
\usepackage[plainpages=false,colorlinks,linkcolor=bleu_sombre,citecolor=rouge_sombre,urlcolor=vert_sombre,breaklinks]{hyperref}

\usepackage[english]{babel}

\usepackage{amsmath,amssymb,amsthm,graphicx,amsfonts,url,color,enumerate,dsfont,stmaryrd,mathabx}

\theoremstyle{plain}
\newtheorem{theorem}{{Theorem}}[section] 
\newtheorem*{theorem*}{{Theorem}}
\newtheorem{proposition}[theorem]{Proposition}
\newtheorem*{proposition*}{Proposition}

\newtheorem*{corollary*}{Corollary}
\newtheorem{lemma}[theorem]{Lemma}
\newtheorem*{lemma*}{Lemma}

\theoremstyle{definition}

\newtheorem*{definition*}{Definition}

\theoremstyle{remark}

\makeatletter

\@addtoreset{equation}{section}  
\makeatother

\newcommand {\limt}[2]{\xrightarrow[#1 \to #2]{}}

\newcommand{\abs}[1]{\left\vert #1\right\vert}        
\newcommand{\nr}[1]{\left\Vert #1\right\Vert}         
\newcommand{\innp}[2]{\left< #1 , #2 \right>}         
\newcommand{\pppg}[1] {\left< #1 \right>}
\newcommand{\set}[1]{\left\{ #1 \right\}}		

\renewcommand{\leq}{\leqslant}	\renewcommand{\geq}{\geqslant}

\newcommand{\inv}{^{-1}}

\newcommand{\littleo}[2]{\mathop{o}\limits_{#1 \to #2}}
\newcommand{\1}{\mathds 1}
\newcommand{\st}{\,:\,}
\newcommand{\restr}[2]{\left.#1\right|_{#2}}         
\renewcommand{\Re}{\mathsf{Re}}        
\renewcommand{\Im}{\mathsf{Im}}  
\newcommand{\trsp}{^{\intercal}} 
\newcommand{\supp}{\mathsf{supp}}

\newcommand{\Dom}{\mathsf{Dom}}

\newcommand{\Id}{\mathsf{Id}} 

\newcommand{\divg}{\mathop{\rm{div}}\nolimits}
\newcommand{\loc}{{\mathsf{loc}}}

\newcommand{\R}{\mathbb{R}}		\newcommand{\C}{\mathbb{C}}
\newcommand{\N}{\mathbb{N}}	\newcommand{\Z}{\mathbb{Z}}	
    
\newcommand{\Sph}{\mathbb{S}}

\renewcommand{\a}{\alpha}\renewcommand{\b}{\beta}\newcommand{\g}{\gamma}\renewcommand{\d}{\delta}\newcommand{\D}{\Delta}\newcommand{\e}{\varepsilon}\newcommand{\z}{\zeta} \renewcommand{\th}{\theta}\newcommand{\Th}{\Theta}\renewcommand{\k}{\kappa}\renewcommand{\l}{\lambda}\newcommand{\x}{\xi}\newcommand{\s}{\sigma}\renewcommand{\t}{\tau}\newcommand{\f}{\varphi}\newcommand{\vf}{\phi}\newcommand{\h}{\chi}\renewcommand{\o}{\omega}

\newcommand{\Ac}{{\mathcal A}}\newcommand{\Gc}{{\mathcal G}}\newcommand{\Hc}{{\mathcal H}}\newcommand{\Kc}{{\mathcal K}}\newcommand{\Lc}{{\mathcal L}}\newcommand{\Rc}{{\mathcal R}}\newcommand{\Sc}{{\mathcal S}}\newcommand{\Tc}{{\mathcal T}}

\newcommand{\stepp}{\noindent {\bf $\bullet$}\quad }

\newcommand{\detail}[1]
{
}

\usepackage{mathrsfs}

\begin{document}

\newcommand{\HH}{\mathcal H}\newcommand{\LL}{\mathcal L}\newcommand{\EE}{\mathscr E}

\newcommand{\Opw}{{\mathop{\rm{Op}}}_h^w}
\newcommand{\Opwn}{{\mathop{\rm{Op}}}_{h_n}^w}

\title[The Klein-Gordon equation with highly oscillating damping]{Energy decay for the Klein-Gordon equation with highly oscillating damping}

\author{Julien Royer}
\address{Institut de Math\'ematiques de Toulouse, Universit\'e Toulouse 3, 118 route de Narbonne - F31062 Toulouse c\'edex 9, France}
\email{julien.royer@math.univ-toulouse.fr}

\begin{abstract}
We consider the free Klein-Gordon equation with periodic damping. We show on this simple model that if the usual geometric condition holds then the decay of the energy is uniform with respect to the oscillations of the damping, and in particular the size of the derivatives do not play any role. We also show that without geometric condition the polynomial decay of the energy is even slightly better for a highly oscillating damping. To prove these estimates we provide a parameter dependent version of well known results of semigroup theory.
\end{abstract}


\maketitle

\section{Introduction and statements of the main results}

Let $d \geq 1$ and $m> 0$. For $(u_0,u_1) \in H^1(\R^d)\times L^2(\R^d)$ we consider on $\R^d$ the damped Klein-Gordon equation
\begin{equation} \label{KG}
\begin{cases}
\partial_t^2 u  -\D u + mu + a_\eta \partial_t u = 0, \quad t \geq 0,\\
(u,\partial_t u)|_{t=0} = (u_0,u_1).
\end{cases}
\end{equation}
For the damping term we consider on $\R^d$ a continuous and $\Z^d$-periodic function $a$ which takes non-negative values and is not identically zero. Then for $\eta \geq 1$ and $x \in \R^d$ we define the absorption index 
\[
a_\eta(x) := a(\eta x).
\]

In this note we are interested in the decay of the energy of the solution $u$. It is defined by 
\[
E(t) := \nr{u(t)}_{H^1}^2 + \nr{\partial_t u(t)}_{L^2}^2,
\]
where $H^1(\R^d)$ is endowed with the norm given by 
\begin{equation} \label{def-norme-H1}
\nr{u}_{H^1}^2 := \nr{\nabla u}_{L^2}^2 + m \nr{u}_{L^2}^2.
\end{equation}
This energy in non-increasing. More precisely, for $t_1 \leq t_2$ we have 
\[
E(t_2) - E(t_1) = -2 \int_{t_1}^{t_2} \int_{\R^d} a_\eta(x) \abs{\partial_t u(t,x)}^2 \, dx \, dt \leq 0.
\]
It is known (see \cite{BurqJo} and references therein) that for $\eta = 1$ the decay is uniform and hence exponential with respect to the initial energy under the so-called Geometric Control Condition. Here, with the free Laplacian, this assumption is that there exist $T > 0$ and $\a > 0$ such that for all $(x,\x) \in \R^d \times \Sph^{d-1}$ we have 
\begin{equation*} 
\pppg {a}_T (x,\x) := \frac 1 T \int_0^T a(x+2t\x) \, dt \geq \a.
\end{equation*}
It is not difficult to check that if this holds for $a$, then it also holds for $a_\eta$ for any $\eta \geq 1$, with constants $T$ and $\a$ which do not depend on $\eta$:
\begin{equation} \label{hyp-GCC}
\exists T > 0, \exists \a > 0, \forall \eta \geq 1, \quad \pppg {a_\eta}_T \geq \a \quad \text{ on } \R^{d} \times \Sph^{d-1}.
\end{equation}

      \detail
      {
      Let $\eta \geq 1$. Let $k \in \N^*$ be such that $\eta T - k T  \in [0, T[$. Then we have 
      \begin{align*}
      \frac 1 T \int_0^T a_\eta (x + 2t\x) \, dt = \frac 1 T \int_0^T a (\eta x + 2 \eta t \x) \, dt = \frac 1 {\eta T} \int_0^{\eta T} a (\eta x + 2s\x) \, ds \geq \frac 1 {(k+1) T} \int_0^{kT} a(\eta x + 2s\x) \, ds \geq \frac \a 2.
      \end{align*}
      }

However, in all the results about uniform energy decay for the damped Klein-Gordon (or wave) equation, some bounds are required for the variations of the absorption index. This rises the natural question wether the exponential decay of the energy $E(t)$ is uniform with respect to $\eta \geq 1$.

\begin{theorem} \label{th-unif-decay}
Assume that the damping condition \eqref{hyp-GCC} holds. Then there exist $\g > 0$ and $C > 0$ such that for $\eta \geq 1$, $(u_0,u_1) \in H^1(\R^d) \times L^2(\R^d)$ and $t \geq 0$ we have 
\[
\nr{u(t)}_{H^1} + \nr{\partial_t u(t)}_{L^2}  \leq C e^{-\g t} \big( \nr{u_0}_{H^1} + \nr{u_1}_{L^2} \big),
\]
where $u$ is the solution of \eqref{KG}.
\end{theorem}

This estimate essentially depends on the contribution of high frequencies. To prove such a result, it is standard to use semiclassical analysis. It is efficient but, on the other hand, it requires a lot of regularity. It is usual to replace the absorption coefficient $a$ by a smooth symbol $\tilde a$ such that $0 \leq \tilde a \leq a$ and $\tilde a$ still satisfies \eqref{hyp-GCC}, possibly with a different $\a$. This idea was already used in \cite{royer-mourre} but the first two radial derivatives of $a$ had to be bounded. This was also used in \cite{BurqJo} but, again, a uniformity on the derivatives of $\tilde a$ was required, so $a$ was assumed to be uniformly continuous. Since the family $(a_\eta)_{\eta \geq 1}$ is not uniformly equicontinous, we cannot prove Theorem \ref{th-unif-decay} with the results of \cite{BurqJo} (see the counter-example of Figure 4.a therein). 

The purpose of this note is to emphasize on a model case that the oscillations of the damping should not play a crucial role in the energy decay of the wave.

For the proof, we will use the same kind of ideas as in \cite{BurqJo} and track (on our periodic setting) the role played by the frequency $\eta$ of the damping.\\

In Theorem \ref{th-unif-decay} we have discussed the energy decay under the damping condition \eqref{hyp-GCC}. It is known that we cannot have uniform decay of the energy without this assumption. However, for a fixed periodic damping, it is proved in \cite{Wunsch} that without any geometric condition we have at least a polynomial decay (with loss of regularity). Here, we prove that this decay is uniform with respect to $\eta$, and moreover the loss of regularity is weaker for the highly oscillating damping. 

This phenomenon is natural. Indeed, for large $\eta$ the damping region becomes in some sense more uniformly distributed in $\R^d$, so even if the average strength of the damping does not depend on $\eta$, and even if \eqref{hyp-GCC} still does not hold for large $\eta$, the distance between undamped classical rays and the damping region gets smaller, so the phenomemon that a high frequency wave approximately following such a ray 
does not see the damping only appears for larger and larger frequencies.

\begin{theorem} \label{th-no-GCC}
There exists $c > 0$ such that for all $\eta \geq 1$, $(u_0,u_1) \in H^2(\R^d) \times H^1(\R^d)$ and $t \geq 0$ we have 
\[
\nr{u(t)}_{H^1} + \nr{\partial_t u(t)}_{L^2} \leq \frac c {\sqrt {1+t}} \left(\nr{u_0}_{H^1} + \nr{u_1}_{L^2} + \frac {\nr{\D u_0}_{L^2} + \nr{\nabla u_1}_{L^2}}{\eta^2} \right),
\]
where $u$ is the solution of \eqref{KG}.
\end{theorem}

For simplicity we have assumed that $a$ is at least continuous, but this is not necessary. For Theorem \ref{th-no-GCC} it is enough to assume that $a$ is bounded and that for some open and $\Z^d$-periodic subset $\o$ of $\R^d$ and $\a_0 > 0$ we have 
\begin{equation} \label{hyp-minor-a}
a \geq \a_0 \1_{\o}.
\end{equation}

For Theorem \ref{th-unif-decay} the assumption is that $a$ is bounded and there exists $\tilde a \in C^\infty(\R^d)$ such that $0 \leq \tilde a \leq a$ and \eqref{hyp-GCC} holds with $a$ replaced by $\tilde a$. As explained in \cite{BurqJo}, this is in particular the case if $a$ is uniformly continuous (for instance, if $a$ is continuous and periodic). We recall that the main point here is that even with a smooth absorption index $a$ the rescaled version $a_\eta$ has derivatives which are not uniformly bounded in $\eta$.\\

This note is organized as follows. In Section \ref{sec-resolvent-to-time} we show how Theorems \ref{th-unif-decay} and \ref{th-no-GCC} are deduced from corresponding resolvent estimates in the energy space. In Section \ref{sec-res-from-L2-to-Hc} we show that these resolvent estimates are in turn consequences of resolvent estimates in the physical space. And finally, in Section \ref{sec-resolvent-L2} we prove these resolvent estimates for a family of Schr\"odinger type operators on $L^2(\R^d)$.

\section{From resolvent estimates to the energy decay} \label{sec-resolvent-to-time}

As usual for the Klein-Gordon equation, we rewrite \eqref{KG} as a first order Cauchy problem in the energy space. We set 
\[
\HH := H^1(\R^d) \times L^2(\R^d),
\]
endowed with the product norm (recall that the norm on $H^1(\R^d)$ is as given by \eqref{def-norme-H1}).
Then, on $\HH$, we consider for $\eta \geq 1$ the operator 
\[
\Ac_\eta := 
\begin{pmatrix}
0 & 1 \\
\D - m  & -a_\eta
\end{pmatrix}
\]
with domain (independent of $\eta$)
\[
\Dom(\Ac) := H^2(\R^d) \times H^1(\R^d).
\]

Let $U_0 = (u_0,u_1) \in \Dom(\Ac)$. Then $u$ is a solution of \eqref{KG} if and only if $U : t \mapsto (u(t),\partial_t u(t))$ is a solution of 
\begin{equation} \label{KG-Ac}
\begin{cases}
\partial_t U - \Ac_\eta U = 0,\\
U(0) = U_0.
\end{cases}
\end{equation}
We will check in Proposition \ref{prop-Ac-diss} that $\Ac_\eta$ generates a contractions semigroup on $\HH$. Then, in this setting, Theorem \ref{th-unif-decay} reads
\begin{equation} \label{decay-HH}
\exists C \geq 0, \exists \gamma > 0, \forall \eta \geq 1, \forall U_0 \in \HH, \forall t \geq 0, \quad \nr{e^{-it\Ac_\eta} U_0}_{\HH} \leq C e^{-\g t} \nr{U_0}_{\HH}.
\end{equation}
And Theorem \ref{th-no-GCC} can be rewritten as
\begin{equation} \label{decay-HH-noGCC}
\exists c \geq 0, \forall \eta \geq 1, \forall U_0 \in \Dom(\Ac), \forall t \geq 0, \quad \nr{e^{-it\Ac_\eta} U_0}_{\HH} \leq \frac c {\sqrt{1+t}} \left( \nr{U_0}_\HH + \frac {\nr{\Ac_\eta U_0}_{\HH}}{\eta^2} \right),
\end{equation}

We are going to prove these estimates from a spectral point of view. More precisely, we will use the following standard results of semigroup theory to deduce \eqref{decay-HH} and \eqref{decay-HH-noGCC} from estimates for the resolvent of $\Ac_\eta$.

\begin{theorem} \label{th-GPH-BT}
Let $\Kc$ be a Hilbert space and let $G$ be an operator on $\Kc$ generating a bounded $C^0$-semigroup $(e^{tG})_{t \geq 0}$. We set 
\[
M = \sup_{t \geq 0} \nr{e^{tG}}_{\Lc(\Kc)},
\]
where $\Lc(\Kc)$ is the space of bounded operators on $\Kc$. Assume that the resolvent set of $G$ contains the imaginary axis.
\begin{enumerate}[(i)]
\item If there exists $C_1 > 0$ such that for all $\t \in \R$ we have 
\[
\nr{(G+i\t)\inv}_{\Lc(\Kc)} \leq C_1,
\]
then there exist $C > 0$ and $\gamma > 0$ which only depend on $C_1$ and $M$ such that for all $t \geq 0$ we have 
\[
\nr{e^{tG}}_{\Lc(\Kc)} \leq C e^{-\gamma t}.
\]
\item If there exist $\kappa \in \N^*$, $c_1 > 0$ and $\nu \in ]0,1]$ such that for all $\t \in \R$ we have 
\[
\nr{(G+i\t)\inv}_{\Lc(\Kc)} \leq c_1 (1+ \nu \abs \t)^{\kappa},
\]
then there exists $c > 0$ which only depends on $c_1$ and $M$ such that for all $t \geq 0$ we have 
\[
\nr{e^{tG}(1+ \nu G)\inv}_{\Lc(\Kc)} \leq c \pppg t^{-\frac 1 \kappa},
\]
where $\pppg \cdot$ stands for $(1+\abs \cdot^2)^{\frac 12}$. 
\end{enumerate}
\end{theorem}

This first statement is a famous result by L. Gearhart \cite{Gearhart78} and J. Pr\"uss \cite{Pruss84} (see also F. Huang \cite{Huang85}). The second statement is due to A. Borichev and Y. Tomilov \cite{BorichevTo10}. Here we recall a proof to check the dependence with respect to the different parameters.

\begin{proof}
\stepp Let $B$ be a bounded operator which commutes with $G$ and such that 
\[
\b := \sup_{\t \in \R} \nr{(G+i\t)^{-1} B}_{\Lc(\Kc)} < +\infty.
\]
The spectrum of $G$ is a subset of the left half-plane and for $\e > 0$ and $\t \in \R$ we have 
\begin{equation} \label{estim-res-G}
\nr{(G-(\e-i\t))\inv}_{\Lc(\Kc)} \leq \frac M \e
\end{equation}
(see for instance Corollary II.1.11 in \cite{engel}). Hence, by the resolvent identity, we have for $\e \in ]0,1]$
\begin{eqnarray} \label{estim-resG-B}
\lefteqn{\nr{(G-(\e-i\t))\inv B}_{\Lc(\Kc)}}\\
\nonumber
&& \leq \nr{(G+i\t)\inv B}_{\Lc(\Kc)} + \e \nr{(G-(\e-i\t))\inv}_{\Lc(\Kc)}\nr{(G+i\t)\inv B}_{\Lc(\Kc)}\\
\nonumber
&& \leq (1+M)\b.
\end{eqnarray}

\stepp 
We define on $\Kc \times \Kc$ the operator 
\[
\Gc = 
\begin{pmatrix}
G & B \\
0 & G
\end{pmatrix},
\]
with domain $\Dom(\Gc) := \Dom(G) \times \Dom(G)$. We can check that $\Gc$ has the same spectrum as $G$, and for $z$ in their common resolvent set we have 
\[
(\Gc - z)\inv =
\begin{pmatrix}
(G-z)\inv & -(G-z)^{-2} B \\
0 & (G-z)\inv
\end{pmatrix}.
\]
Moreover $\Gc$ generates the $C^0$-semigroup given by  
\[
e^{t\Gc} =
\begin{pmatrix}
e^{tG} & t e^{tG} B \\
0 & e^{tG}
\end{pmatrix}, \quad t \geq 0.
\]

\stepp 
Let $\e > 0$ and $\f \in \Kc$. Since $\t \mapsto (G - (\e-i\t))\inv \f$ is the inverse Fourier transform of $t \mapsto - \1_{\R_+}(t) e^{-t\e} e^{tG} \f$, we obtain by the Parseval identity
\[
\int_\R \nr{(G - (\e-i\t))\inv \f}_\Kc^2  \, d\t \leq \frac {\pi M^2 }{\e}\nr{\f}_\Kc^2. 
\]
Since $B$ commutes with $G$ we also have by \eqref{estim-resG-B}
\begin{multline*}
\int_\R \nr{(G - (\e-i\t))^{-2} B \f}_\Kc^2  \, d\t \\
\leq \int_\R \nr{(G - (\e-i\t))\inv B}_{\Lc(\Kc)}^2 \nr{(G-(\e-i\t))\inv \f}_\Kc^2  \, d\t \leq \frac {\pi (1+M)^2M^2 \b^2}{\e}\nr{\f}_\Kc^2.
\end{multline*}
Then, for $\e > 0$ and $\Phi = (\f_1,\f_2) \in \Kc \times \Kc$,
\begin{eqnarray} \label{estim-res-Gc}
\lefteqn{\int_{\R} \nr{(\Gc-(\e-i\t))\inv \Phi}_{\Kc \times \Kc}^2 \, d\t}\\
\nonumber
&& \leq \int_{\R}  \left(\nr{(G - (\e-i\t))\inv \f_1}_\Kc^2 + \nr{(G-(\e-i\t))^{-2} B \f_2}_\Kc^2 +  \nr{(G - (\e-i\t))\inv \f_2}_\Kc^2 \right) \, d\t\\
\nonumber
&& \leq \frac {c_{M,\beta}} {\e} \nr{\Phi}_{\Kc\times \Kc}^2,
\end{eqnarray}
and we have a similar estimate with $(\Gc-(\e-i\t))\inv$ replaced by $(\Gc^*-(\e+i\t))\inv$. Then for $t > 0$, $\e > 0$ and $\Phi, \Psi \in \Kc \times \Kc$ we use the identity
\[
\innp{e^{t\Gc} \Phi}{\Psi}_{\Kc\times \Kc} = \frac 1 {2i\pi t} \int_{\t \in \R} e^{(\e-i\t)t} \innp{(\Gc-(\e-i\t))^{-2} \Phi}{\Psi}_{\Kc\times \Kc} \, d\t
\]
(see for instance Corollary III.5.16 in \cite{engel}). Applied with $\e = 1/t$ this gives 
\begin{align*}
\abs{\innp{e^{t\Gc} \Phi}{\Psi}}
& \leq \frac {\e e}{2\pi} \int_{\t \in \R}  \nr{(\Gc-(\e-i\t))\inv \Phi} \nr{(\Gc^*-(\e+i\t))\inv \Psi} \, d\t .
\end{align*}
By the Cauchy-Schwarz inequality and \eqref{estim-res-Gc} we obtain 
\[
\abs{\innp{e^{t\Gc} \Phi}{\Psi}} \leq \tilde c_{M,\b} \nr{\Phi}_{\Kc \times \Kc} \nr{\Psi}_{\Kc \times \Kc}.
\]
This gives a bound for $e^{t\Gc}$, and in particular there exists $C_{M,\b} > 0$ such that for all $t \geq 0$ we have
\begin{equation} \label{estim-etG-B}
\nr{e^{tG} B}_{\Lc(\Kc)} \leq \frac {C_{M,\beta}} {\pppg t}.
\end{equation}

	\detail 
	{
	\begin{align*}
	- (G-(\e-i\t)) \int_0^{+\infty} e^{it\t} e^{-t\e} e^{tG} \f \, dt = - (G-(\e-i\t)) \int_0^{+\infty} e^{t(G-(\e-i\t))} \f \, dt = - \int_0^{+\infty} \frac d {dt} e^{t(G-(\e-i\t))} \f \, dt = \f.
	\end{align*}
	\begin{align*}
	\int_\R \nr{(G - (\e-i\t))\inv \f}_\Kc^2 \, d\t
	& = 2 \pi \int_0^{+\infty} e^{-2\e t} \nr{e^{tG} \f}^2 \, dt
	& \leq 2 \pi M^2 \nr{\f}^2 \int_0^{+\infty} e^{-2 \e t} \, dt 
	& \leq \frac {\pi M^2}{\e} \nr{\f}_\Kc^2. 
	\end{align*}
	}

\stepp We prove the second statement of Theorem \ref{th-GPH-BT}. Since the resolvent is continuous on the imaginary axis, it is bounded on any compact subset. Thus it is enough to prove the estimate for $\abs \tau \geq 1$. By the resolvent identity we can prove by induction on $\kappa \in \N^*$ that 
\[
(G+i\tau)\inv (G+\nu\inv)^{-\kappa} = \frac {(G+i\tau)\inv}{(\nu\inv-i\tau)^\kappa} - \sum_{k=1}^{\kappa} \frac {(G+\nu\inv)^{-(\kappa+1-k)}}{(\nu\inv-i\tau)^k}.
\]
Then
\begin{align*}
\nr{(G+i\t)\inv (\nu G + 1)^{-\kappa}}
& = \frac 1 {\nu^\kappa} \nr{(G+i\t)\inv (G + \nu\inv)^{-\kappa}}\\
& \leq \frac {\nr{(G+i\tau)\inv}} {(1+\nu^2 \tau^2)^{\frac \kappa 2}} + \sum_{k=1}^\kappa \frac {\nr{(G+\nu\inv)\inv}^{\kappa+1-k}}{\nu^{\kappa-k} (1+\nu^2\tau^2)^{\frac k 2}}.
\end{align*}
By assumption on the resolvent and by \eqref{estim-res-G} this is uniformly bounded by a constant which only depends on $M$ and $c_1$.  Then by \eqref{estim-etG-B} applied with $B = (\nu G + 1)^{-\kappa}$ there exists $\tilde c$ which only depends on $c_1$ and $M$ such that for all $t \geq 0$ we have
\[
\nr{e^{tG} (\nu G + 1)^{-\kappa}}_{\Lc(\Kc)} \lesssim \frac {\tilde c} {\pppg t}.
\]
Finally, we can follow the proof of \cite[Proposition 3.1]{BatkaiEngPruSch06} to conclude.

\stepp Now we turn to the proof of the first statement of Theorem \ref{th-GPH-BT}. By \eqref{estim-etG-B} applied with $B=\Id_\Kc$ there exists $\tilde C \geq 0$ (which only depends on $C_1$ and $M$) such that 
\[
\nr{e^{tG}}_{\Lc(\Kc)} \leq \frac {\tilde C}{t}.
\]
In particular for $T = 2 \tilde C$ we get $\nr{e^{TG}} \leq 1/2$. Then for $t \geq T$ we denote by $k$ the integer part of $t/T$ and write 
\[
\nr{e^{tG}} \leq \nr{e^{TG}}^k \nr{e^{(t-kT)G}} \leq \frac M {2^k} \leq 2 M  e^{-\frac {t\ln(2)}T}.
\]
The proof is complete.
\end{proof}

Thus, in order to prove \eqref{decay-HH} and \eqref{decay-HH-noGCC} (and hence Theorems \ref{th-unif-decay} and \ref{th-no-GCC}), it is enough to prove the following resolvent estimates for $\Ac_\eta$:

\begin{theorem} \label{th-res-Ac}
\begin{enumerate}[(i)]
\item Let $\eta \geq 1$. Then $\Ac_\eta$ generates a bounded $C^0$-semigroup and its resolvent set contains the imaginary axis.
\item There exists $c_1 > 0$ such that for $\eta \geq 1$ and $\t \in \R$ we have 
\[
\nr{(\Ac_\eta + i\tau)\inv}_{\Lc(\HH)} \leq c_1 \left(1+\frac {\abs{\t}}{\eta^2} \right)^2.
\]
\item If moreover \eqref{hyp-GCC} holds, then there exists $C_1 > 0$ such that for $\eta \geq 1$ and $\t \in \R$ we have 
\[
\nr{(\Ac_\eta + i\tau)\inv}_{\Lc(\HH)} \leq C_1.
\]
\end{enumerate}
\end{theorem}

\section{Resolvent estimates in the energy space} \label{sec-res-from-L2-to-Hc}

In this section we discuss the proof of Theorem \ref{th-res-Ac}. Introducing the wave operator $\Ac_\eta$ on $\HH$ was useful to apply the general results of the semigroup theory. However, to prove concrete resolvent estimates we have to go back to the analysis of Schr\"odinger operators on $L^2(\R^d)$.

\begin{proposition} \label{prop-res-Ac-Rz}
Let $\eta \geq 1$ and $z \in \C$. Then $(\Ac_\eta +iz)$ is invertible with bounded inverse on $\Hc$ if and only if the operator
$
{\big(-\D + m -iza_\eta -z^2\big)}
$
is invertible with inverse bounded on $L^2(\R^d)$, and in this case we have 
\begin{equation} \label{expr-res-Ac}
(\Ac_\eta +iz)\inv =
\begin{pmatrix}
R^\eta(z) (-a_\eta + iz) & -R^\eta(z) \\
1 + R^\eta(z) (iz a_\eta+z^2) & i z R^\eta(z)
\end{pmatrix},
\end{equation}
where we have set 
\[
R^\eta(z) = \big(-\D + m -iza_\eta -z^2\big)\inv.
\]
\end{proposition}

\begin{proof}
Assume that $R^\eta(z)$ is well defined. It is a bounded operator from $L^2(\R^d)$ to $H^2(\R^d)$, so the right-hand side of \eqref{expr-res-Ac} defines a bounded operator from $\HH$ to $\Dom(\Ac)$. Then we can check by direct computation that it is a bounded inverse for $(\Ac_\eta + iz)$. Conversely, assume that $-iz$ belongs to the resolvent set of $\Ac_\eta$. For $g \in L^2(\R^d)$ we define $Rg$ as the first component of $(\Ac_\eta + iz)\inv G$, for $G = (0,-g) \in \HH$. This defines a bounded operator $R$ from $L^2(\R^d)$ to $H^2(\R^d)$ and we can check, again by direct computation, that it is an inverse for ${(-\D + m - iza_\eta - z^2)}$.
\end{proof}

We begin the proof of Theorem \ref{th-res-Ac} with the statement that $\Ac_\eta$ generates a bounded $C^0$-semigroup. For this we prove that $\Ac_\eta$ is \textsf{m}-dissipative. By the usual Lummer-Phillips Theorem, this ensures that $\Ac_\eta$ generates for all $\eta \geq 1$ a contractions semigroup. We recall that an operator $T$ with domain $\Dom(T)$ on a Hilbert space $\Kc$ is said to be dissipative if for all $\f \in \Dom(T)$ we have 
\[
\Re \innp{T\f}{\f} \leq 0.
\]
Moreover $T$ is said to be \textsf{m}-dissipative if some (and hence any) $\z \in \C$ with $\Re(\z) > 0$ belongs to the resolvent set of $T$.

\begin{proposition} \label{prop-Ac-diss}
For all $\eta \geq 1$ the operator $\Ac_\eta$ is \textsf{m}-dissipative on $\HH$.
\end{proposition}

\begin{proof}
Let $U = (u,v) \in \Dom(\Ac)$. We have 
\[
\innp{\Ac_\eta U}{U}_\HH = \innp{\nabla v}{\nabla u}_{L^2} + m \innp{v}{u}_{L^2} + \innp{(\D - m) u}{v}_{L^2} - \innp{a_\eta v}{v}_{L^2},
\]
so 
\[
\Re \innp{\Ac_\eta U}{U}_\HH = - \innp{a_\eta v}{v}_{L^2} \leq 0.
\]
This proves that $\Ac_\eta$ is dissipative. On the other hand, the operator $-\D + m + a_\eta + 1$ is self-adjoint and bounded below by $m+1$. In particular it is invertible with bounded inverse on $L^2(\R^d)$. By Proposition \ref{prop-res-Ac-Rz}, this implies that 1 is in the resolvent set of $\Ac_\eta$, hence $\Ac_\eta$ is \textsf{m}-dissipative.
\end{proof}

The estimates for the resolvent of $\Ac_\eta$ will be deduced from estimates of the ``resolvent'' $R^\eta(\t)$ defined in Proposition \ref{prop-res-Ac-Rz}. To work with a fixed damping, we first rescale the problem. For $\eta \geq 1$, $u \in L^2$ and $x \in \R^d$ we set 
\[
(\Th_\eta u)(x) = \eta^{\frac d 2} u(\eta x).
\]
This defines a unitary operator $\Th_\eta$ on $L^2(\R^d)$ and we have 
\[
\Th_\eta\inv \big(-\D + m -i \t a_\eta - \t^2\big) \Th_\eta = \big(-\eta^2 \D + m -i \t a - \t^2\big).
\]
In particular the operator $\big(-\eta^2 \D + m -i \t a - \t^2\big)$ has an inverse bounded on $L^2(\R^d)$ if and only if $R^\eta(\t)$ is well defined, and in this case, if we set 
\begin{equation} \label{def-Reta}
R_\eta(\t) = \big(-\eta^2 \D + m -i \t a - \t^2\big)\inv,
\end{equation}
then 
\begin{equation} \label{eq-norm-res-eta}
\nr{R^\eta(\t)}_{\Lc(L^2)} = \nr{R_\eta(\t)}_{\Lc(L^2)}.
\end{equation}
Finally, Theorem \ref{th-res-Ac} will be a consequence on the following estimates on $R_\eta(\t)$:

\begin{proposition} \label{prop-res-eta}
\begin{enumerate}[(i)]
\item For all $\eta \geq 1$ and $\t \in \R$ the operator $\big(-\eta^2 \D + m -i \t a - \t^2\big)$ is invertible with bounded inverse on $L^2(\R^d)$.
\item There exists $c_2 > 0$ such that for $\eta \geq 1$ and $\t \in \R$ we have 
\[
\nr{R_\eta(\t)}_{\Lc(L^2)} \leq \frac {c_2}{\pppg \tau} \left(1 + \frac {\abs \t} {\eta^2} \right)^2.
\]
\item If \eqref{hyp-GCC} holds then there exists $C_2 > 0$ such that for $\eta \geq 1$ and $\t \in \R$ we have 
\[
\nr{R_\eta(\t)}_{\Lc(L^2)} \leq \frac {C_2} {\pppg \t}.
\]
\end{enumerate}
\end{proposition}

The proof of Proposition \ref{prop-res-eta} is postponed to the following section. Here we show that it indeed implies Theorem \ref{th-res-Ac}.

\begin{proof}[Proof of Theorem \ref{th-res-Ac}, assuming Proposition \ref{prop-res-eta}]
\stepp Since $R_\eta(\t)$ is well defined for all $\eta \geq 1$ and $\t \in \R$, this is also the case for $R^\eta(\t)$ and hence for $(\Ac_\eta +i\t)\inv$ by Proposition \ref{prop-res-Ac-Rz}. Moreover, by \eqref{eq-norm-res-eta}, the estimates given for $R_\eta(\t)$ also hold for $R^\eta(\t)$. 

\stepp Assume that for $\eta \geq 1$ and $\t \in \R$ we have 
\[
\nr{R_\eta(\t)} \leq \frac {\k(\eta,\t)}{\pppg \t},
\]
where $\k$ is bounded below by a positive constant. For $\eta \geq 1$, $\t \in \R$ and $u$ in the Schwartz space $\Sc(\R^d)$ we have 
\begin{align*}
\nr{\nabla R^\eta(\t) u}_{L^2}^2
 = \innp{u}{R^\eta(\t) u} + \innp{(-m + i\t a_\eta + \t^2) R^{\eta}(\t)u}{R^\eta(\t) u}
 \lesssim \k(\eta,\t)^2 \nr{u}_{L^2}^2.
\end{align*}
This proves that 
\[
\nr{R^\eta(\t)}_{\Lc(L^2,H^1)} \lesssim \k(\eta,\tau).
\]
By duality we also have 
\[
\nr{R^\eta(\t)}_{\Lc(H\inv,L^2)} = \nr{R^\eta(\t)^*}_{\Lc(L^2,H^1)} = \nr{R^\eta(-\t)}_{\Lc(L^2,H^1)} \lesssim \k(\eta,\tau).
\]
Then, as above,
\begin{align*}
\nr{\nabla R^\eta(\t) \nabla u}_{L^2}^2
& = \innp{\nabla u}{R^\eta(\t) \nabla u} + \innp{(-m + i\t a_\eta + \t^2) R^{\eta}(\t) \nabla u}{R^\eta(\t) \nabla u}\\
& \lesssim \nr{\nabla R^\eta(\t) \nabla u}_{L^2} \nr{u}_{L^2} +  \pppg \tau^2 \k(\eta,\tau)^2.
\end{align*}
This yields 
\[
\nr{R^\eta(\t)}_{\Lc(H\inv,H^1)} \lesssim \pppg \t \k(\eta,\tau).
\]

\stepp Let $U = (u,v) \in \Sc \times \Sc$. By Proposition \ref{prop-res-Ac-Rz} we have 
\begin{align*}
\nr{(\Ac_\eta + i\t)\inv U}_{\HH}
& \lesssim \nr{R^\eta(\t) (-a_\eta+i\t)u}_{H^1} + \nr{R^\eta(\t)v}_{H^1}\\
& + \nr{u + R^\eta(\t)(i\t a_\eta + \t^2)u}_{L^2} + \nr{\t R^\eta(\t)v}_{L^2}.
\end{align*}
First, for $\abs \t \geq 1$,
\[
\nr{R^\eta(\t) (-a_\eta+i\t)u}_{H^1} = \frac {\nr{u}_{H^1}} \t  + \frac 1 \t \nr{R^\eta(\t) (\D-m) u}_{H^1} \lesssim \k(\eta,\tau) \nr{u}_{H^1}.
\]
This estimate also holds for $\abs \t \leq 1$ and, similarly,
\[
\nr{u + R^\eta(\t)(i\t a_\eta + \t^2)u}_{L^2} = \nr{R^\eta(\t)(-\D+m) u}_{L^2} \lesssim \k(\eta,\tau) \nr{u}_{H^1}.
\]
We also have 
\[
\nr{R^\eta(\t)v}_{H^1} + \nr{\t R^\eta(\t)v}_{L^2} \lesssim \k(\eta,\tau) \nr{v}_{L^2},
\]
so
\[
\nr{(\Ac_\eta + i\t)\inv U}_{\HH} \lesssim \k(\eta,\tau) \left(\nr{u}_{H^1} + \nr{v}_{L^2} \right) \lesssim \k(\eta,\tau) \nr{U}_\HH.
\]
Thus the second and third statements of Theorem \ref{th-res-Ac} follow from the corresponding statements of Proposition \ref{prop-res-eta}.
\end{proof}

\section{Resolvent estimates for the rescaled operator} \label{sec-resolvent-L2}

In this section we prove Proposition \ref{prop-res-eta}. This will conclude the proof of Theorems \ref{th-unif-decay} and \ref{th-no-GCC}. The three statements of Proposition \ref{prop-res-eta} are proved separately in Propositions \ref{prop-Reta-existence}, \ref{prop-estim-Reta} and \ref{prop-estim-Reta-GCC} below.\\

When working in a periodic setting, it is standard to introduce the Floquet-Bloch decomposition to reduce the problem on $\R^d$ to a family of problems on the torus. Here we use the notation of \cite{JolyRo}. For $u \in \Sc(\R^d)$ and $\s \in \R^d$ we set 
\[
u^\sharp_\s (x) = \sum_{n \in \Z^d} u(x + n) e^{-i(x+n)\cdot \s}.
\]
This defines for all $\s$ a $\Z^d$-periodic function and for $x \in \R^d$ we have 
\[
u(x) = \frac 1 {(2\pi)^d} \int_{\s \in [0,2\pi]^d} e^{ix\cdot \s} u^\sharp_\s(x) \, d\s.
\]
Moreover we have the Parseval identity 
\[
\nr{u}_{L^2}^2 = \frac 1 {(2\pi)^d} \int_{\s \in [0,2\pi]^d} \nr{u_\s^\sharp}_{L^2_\sharp}^2 \, d\s,
\]
where $L^2_\sharp$ is the set of $L^2_\loc$ and $\Z^d$-periodic functions on $\R^d$, endowed with the norm given by
\[
\nr{u_\s^\sharp}_{L^2_\sharp}^2 = \int_{[0,1]^d} \abs{u_\s^\sharp(x)}^2 \, dx.
\]

\begin{proposition} \label{prop-Reta-existence}
For all $\eta \geq 1$ and $\t \in \R$ the operator $\big(-\eta^2 \D + m - i\t a - \t^2 \big)$ has a bounded inverse on $L^2(\R^d)$.
\end{proposition}

\begin{proof} Let $\eta \geq 1$ and $\t \in \R$ be fixed. If $\t = 0$ it is clear that the selfadjoint operator $-\eta^2 \D + m$ is bounded below by $m > 0$ and hence invertible. Now we assume that $\t \neq 0$. For $\s \in \R^d$ we set 
\[
\D_\s = e^{-ix\cdot \s} \D e^{ix\cdot \s} = (\divg + i \s\trsp)  (\nabla + i\s).
\]
and
$P_\s = \big(-\eta^2 \D_\s + m - i\t a - \t^2 \big)$ ($\Dom(P_\s)$ is the set of $H^2_\loc$ and $\Z^d$-periodic functions). Then for $u \in \Sc(\R^d)$ we have 
\[
\big(-\eta^2 \D + m - i\t a - \t^2 \big) u = \frac 1 {(2\pi)^d} \int_{\s \in [0,2\pi]^d} e^{ix\cdot \s} P_\s u^\sharp_\s(x) \, d\s.
\]
Let $\s \in \R^d$. The operator $P_\s$ has nonempty resolvent set and compact resolvent, so its spectrum is given by a sequence of eigenvalues. Let $u \in \Dom(P_\s)$ be such that $P_\s u = 0$. Since 
\[
\nr{\sqrt a u}_{L^2_\sharp}^2 = - \frac {\Im \innp{P_\s u}{u}} \t = 0, 
\]
we get that $u$ vanishes in an open subset of $\R^d$ so, by unique continuation, $u = 0$. Then 0 is not an eigenvalue of $P_\s$, so $P_\s$ is invertible with bounded inverse in $L^2_\sharp$. We set $R_\s = P_\s\inv$. Then for $f \in \Sc(\R^d)$ we set 
\[
R f = \frac 1 {(2\pi)^d} \int_{\s \in [0,2\pi]^d} e^{ix\cdot \s} R_\s u^\sharp_\s(x) \, d\s.
\]
Since $R_\s$ is a continuous function of $\s$, it is bounded on $[0,2\pi]^d$. Then, by the Parseval identity, 
\[
\nr{Rf}_{L^2}^2 = \frac 1 {(2\pi)^d} \int_{\s \in [0,2\pi]^d} \nr{R_\s f_\s^\sharp}_{L^2_\sharp}^2 \, d\s \lesssim \frac 1 {(2\pi)^d} \int_{\s \in [0,2\pi]^d} \nr{ f_\s^\sharp}_{L^2_\sharp}^2 \, d\s = \nr{f}_{L^2}^2.
\]
Thus $R$ defines a bounded operator on $L^2(\R^d)$. Then we check that it is an inverse for $\big(-\eta^2 \D + m - i\t a - \t^2 \big)$ and the proposition is proved.
\end{proof}

Now we turn to the proof of the second statement of Proposition \ref{prop-res-eta}. It relies on the following observability estimate:

\begin{proposition} \label{prop-Wun}
Let $\o$ be a nonempty, open and $\Z^d$-invariant subset of $\R^d$. Then there exists $C > 0$ such that for all $u \in H^2(\R^d)$ and $\l \in \R$ we have 
\[
\nr{u}_{L^2} \leq C \left( \nr{(-\D - \l)u}_{L^2} + \nr{u}_{L^2(\o)} \right).
\]
\end{proposition}

This kind of estimate is a difficult result in general. It is only known in very particular settings (see for instance \cite{Jaffard90,BurqZwo12,AnantharamanLeMa16}). Proposition \ref{prop-Wun} is deduced in \cite{Wunsch} from the case of the torus by means of the Floquet-Bloch decomposition as above. With this proposition in hand, we can prove the following resolvent estimate:

\begin{proposition} \label{prop-estim-Reta}
There exists $c_2 > 0$ such that for all $\eta \geq 1$ and $\t \in \R$ we have 
\[
\nr{R_\eta(\t)}_{\Lc(L^2)} \leq \frac {c_2}{\pppg \tau} \left( 1 + \frac {\abs \t}{\eta^2} \right)^2.
\]
\end{proposition}

\begin{proof} 

We have $\nr{R_\eta(0)} \leq 1/m$, so if $\t_0 > 0$ is such that $\t_0 \nr{a}_\infty + \t_0^2 \leq m/2$ then by a standard perturbation argument we have $\nr{R_\eta(\t)} \leq 2/m$ for all $\t \in [-\t_0,\t_0]$. Thus, in the rest of the proof it is enough to estimate $\nr{R_\eta(\t)}$ for $\abs \t \geq \t_0$. So let $u \in H^2(\R^d)$, $\eta \geq 1$ and $\t \in \R$ with $\abs \t \geq \t_0$. We set 
\begin{equation*} 
f = \big(-\eta^2 \D + m - i \t a - \t^2\big) u.
\end{equation*}
This can be rewriten as 
\[
-\D u - \frac {\t^2-m} {\eta^2} u = \frac{f + i\t a u}{\eta^2}.
\]
By Proposition \ref{prop-Wun} applied with $\o$ given by \eqref{hyp-minor-a} we obtain 
\begin{equation} \label{estim-nr-u}
\nr{u} \lesssim \frac 1 {\eta^2} \left(\nr{f} + \abs \t \nr{au} \right) + \nr{u}_{L^2(\o)} \lesssim  \frac {\nr{f}}{\eta^2} + \left(1+\frac {\abs \t}{\eta^2} \right) \nr{a u}.
\end{equation}
Since $a$ is bounded we have $a \lesssim \sqrt a$, so for any $\e > 0$ we have 
\begin{align*} 
\nr{a u}_{L^2}^2
& \lesssim \nr{\sqrt a u}_{L^2}^2 = - \frac {\Im \innp{f}{u}}{\t} \leq \frac {\nr f_{L^2} \nr u_{L^2}} {\abs \t}\\
& \leq \e^2 \left(1 + \frac {\abs \t}{\eta^2} \right)^{-2} \nr{u}^2 + \left(1 + \frac {\abs \t}{\eta^2} \right)^2 \frac {\nr{f}^2}{4\e^2 \tau^2}.
\end{align*}
Then \eqref{estim-nr-u} gives
\[
\nr{u} \lesssim \e \nr{u} + \frac {C_\e}{\abs \tau} \left(1 + \frac {\abs \t}{\eta^2} \right)^2 \nr{f}.
\]
With $\e > 0$ chosen small enough we get 
\[
\nr{u} \lesssim \frac {1}{\abs \tau} \left(1 + \frac {\abs \t}{\eta^2} \right)^2\nr{f},
\]
which gives the required estimate for $R_\eta(\t)$.
\end{proof}

We finally prove the last statement of Proposition \ref{prop-res-eta}:

\begin{proposition} \label{prop-estim-Reta-GCC}
If the damping condition \eqref{hyp-GCC} holds, then there exists $C_2 > 0$ such that for all $\eta \geq 1$ and $\t \in \R$ we have 
\[
\nr{R_\eta(\t)}_{\Lc(L^2)} \leq \frac {C_2} {\pppg \t}.
\]
\end{proposition}

We notice that as long as $\abs \t$ remains comparable to $\eta^2$ this is a consequence of Proposition \ref{prop-estim-Reta}. Thus, Proposition \ref{prop-estim-Reta-GCC} is only a result about frequencies greater that $\eta^2$.

One of the standard method to prove such a resolvent estimate under a suitable geometric condition about classical trajectories is to use semiclassical analysis (see for instance \cite{zworski} for an introduction to the subject) and, more precisely, the contradiction method of \cite{Lebeau96}. For this, we rewrite the problem in a semiclassical setting. More precisely, Proposition \ref{prop-estim-Reta-GCC} is a consequence of Proposition \ref{prop-estim-Reta} and of the following lemma, applied with $h = \eta / \t$ and $\e = 1/\eta$:

\begin{lemma} \label{lem-high-freq}
There exist $h_0 > 0$ and $C_3 > 0$ such that for $h \in ]0,h_0]$ and $\e \in ]0,1]$ we have 
\[
\nr{\big(-h^2 \D - i \e h a - 1 \big) \inv}_{\Lc(L^2)} \leq \frac {C_3} {\e h}.
\]
\end{lemma}

The difference with the usual high frequency estimates for the damped wave equation is that we make more explicit the dependence with respect to the strength of the damping. For the proof we essentially follow \cite{BurqJo} and check the dependence in $\e$. Notice that up to now we have only used assumption \eqref{hyp-minor-a}. It is only for the proof of Lemma \ref{lem-high-freq} that we need to replace $a$ by a smooth absorption index.

\begin{proof}
\stepp 
We construct an absorption index $a_\infty \in C^\infty(\R^d)$ such that $0 \leq a_\infty \leq a$ and $\pppg{a_\infty}_T \geq \a / 2$ on $\R^d \times \Sph^{d-1}$, where $T > 0$ and $\a > 0$ are given by \eqref{hyp-GCC}. 
For this, we set $\tilde a = {\max(0,a - \a / 4)}$. Since $a$ is continuous and periodic, so is $\tilde a$, and there exists $\d > 0$ such that if $\tilde a(x) > 0$ then $a$ is positive on the ball $B(x,2\d)$. 
Moreover, there exists $a_0 > 0$ such that $a \geq a_0$ on a $\d$-neighborhood of the support of $\tilde a$. On the other hand, since $\tilde a$ is continuous and periodic it is uniformly continuous, so we can choose $\d$ smaller to ensure that $\abs{\tilde a(x_1) - \tilde a(x_2)} \leq \min(a_0,\a/4)$ whenever $\abs{x_1-x_2} \leq \d$. Let $\rho \in C^\infty(\R^d,\R_+)$ be supported in the ball $B(0,\d)$ and of integral 1. We set $a_\infty = \tilde a * \rho$. 
Then $a_\infty$ is smooth and takes non-negative values. It is supported in the $\d$-neighborhood of $\supp(\tilde a)$ and $\nr{\tilde \a - a_\infty}_\infty \leq \min(a_0, \a/4)$ so $a_\infty \leq a$. 
Moreover $\nr{a-a_\infty}_\infty \leq \a/2$ so $\pppg {a_\infty}_T \geq \a / 2$ on $\R^d \times \Sph^{d-1}$. Then, by continuity and periodicity of $a_\infty$, there exists $\e > 0$ such that for $(x,\x) \in \R^{2d}$ with $1-\e \leq \abs \x^2 \leq 1 + \e$ we have
\begin{equation} \label{GCC-inf}
\pppg {a_\infty}_T (x,\x) \geq \frac {\a} 4.
\end{equation}

\stepp Assume by contradiction that the statement of the lemma is wrong. Then we can find sequences $(u_n)_{n \in \N} \in (H^2(\R^d))^\N$, $(h_n)_{n \in \N} \in ]0,1]^\N$ and $(\e_n)_{n \in \N} \in ]0,1]^\N$ such that $h_n \to 0$, $\nr{u_n}_{L^2} = 1$ for all $n \in \N$ and 
\begin{equation} \label{estim-Pn-un}
\nr{(-h_n^2 \D - i \e_n h_n a - 1)u_n}_{L^2} = \littleo n \infty (\e_n h_n).
\end{equation}

\stepp We have 
\begin{equation*}
\innp{a u_n}{u_n} = - \frac 1 {\e_n h_n} \Im \innp{ (-h_n^2 \D - i \e_n h_n a - 1) u_n}{u_n}_{L^2} \limt n \infty 0,
\end{equation*}
and in particular 
\begin{equation} \label{eq-nr-a-un}
\nr{a_\infty u_n}_{L^2} \leq \nr{a u_n}_{L^2} \lesssim \nr{\sqrt a u_n}_{L^2} \limt n \infty 0.
\end{equation}

\stepp For $n \in \N$ we set $P_n = (-h_n^2 \D  -i\e_n h_n a - 1)$. Let $q \in C_b^\infty(\R^{2d},\R)$ (the set of smooth and real valued functions with bounded derivatives). We consider the Weyl quantization of $q$ 
\[
\Opw(q) u (x) = \frac 1 {(2\pi h)^{\frac d 2}} \int_{\R^d} \int_{\R^d} e^{\frac i h \innp{x-y} \x} q\left( \frac {x+y}2,\x \right) u(y) \, dy \, d\x.
\]
We have
\begin{align*}
\innp{\Opwn(\{\x^2,q\}) u_n}{u_n}
& = \frac 1 {h_n} \innp{[-h_n^2 \D ,\Opwn(q)] u_n}{u_n} + O(h_n)\\
& = \frac 1 {h_n} \big(\innp{\Opwn(q) u_n}{P_n u_n} - \innp{P_n u_n}{\Opwn(q) u_n} \big)\\
&\quad  - i \e_n \big(\innp{\Opwn(q)u_n}{au_n} + \innp{a u_n} {\Opwn(q) u_n}\big) + O(h_n),
\end{align*}
so by \eqref{estim-Pn-un} and \eqref{eq-nr-a-un}
\begin{equation} \label{eq-poisson-bracket-1}
\innp{\Opwn(\{\x^2,q\}) u_n}{u_n} \limt n \infty 0.
\end{equation}

\stepp Let $\k \in C_b^\infty(\R^{2d},\R)$ be equal to 0 on $\set{(x,\x) \in \R^{2d} \st \abs{\abs \x^2 - 1} \leq \e_0}$ for some $\e_0 > 0$. Let
\[
\tilde \k : (x,\x) \mapsto \frac {\k(x,\x)}{\abs \x^2 - 1}.
\]
This also defines a symbol in $C_b^\infty(\R^{2d},\R)$ and we have 
\begin{equation} \label{eq-kappa}
\begin{aligned}
\innp{\Opwn(\k) u_n}{u_n}
& = \innp{\Opwn(\tilde \k) (-h_n^2\D - 1) u_n}{u_n} + O(h_n)\\
& = \innp{\Opwn(\tilde \k) P_n u_n}{u_n} + O(h_n)\\
& \limt n \infty 0.
\end{aligned}
\end{equation}

\stepp For $(x,\x) \in \R^{2d}$ we set 
\[
b(x,\x): = \frac 2 T \int_0^T \int_0^t a_\infty (x+2s\x) \, ds \, dt = \frac 2 T \int_0^T (T-s) a_\infty(x+2s\x) \, ds.
\]
We have 
\[
b(x+2\th\x,\x) = \frac 2 T \int_\th^{T+\th} (T-s+\th) a_\infty(x+2s\x) \, ds,
\]
so 
\[
\{\x^2,b\}(x,\x) = \restr{\frac d {d\th} b(x+2\th \x,\x)}{\th = 0} = - 2 a_\infty(x) + 2 \pppg {a_\infty}_T (x,\x).
\]
Let $\h \in C_0^\infty(\R)$ be supported in $]1-\e,1+\e[$ and equal to 1 on a neighborhood of 1 ($\e$ was defined before \eqref{GCC-inf}). For $(x,\x) \in \R^{2d}$ we set 
\[
q(x,\x) := \h(\x^2) e^{b(x,\x)} \geq \h(\x^2).
\]
Then 
\[
\{\x^2,q\} = q \{\x^2,b\}  = 2q\pppg {a_\infty}_T - 2a_\infty q.
\]
By \eqref{GCC-inf} we have
\[
\{\x^2,q\} + 2a_\infty q + \frac \a 2 (1-\h)(\x^2) \geq \frac \a 2.
\]
By the G\aa rding inequality we obtain for $n$ large enough
\begin{align*}
\innp{\Opwn(\{\x^2,q\})u_n}{u_n} 
\geq \frac {\a} 4 - \innp{\Opwn \left(2 a_\infty q + \frac \a 2 (1-\h)(\x^2) \right) u_n}{u_n} 
\end{align*}
With \eqref{eq-nr-a-un} and \eqref{eq-kappa} we get
\[
\liminf_{n \to \infty} \innp{\Opwn(\{\x^2,q\})u_n}{u_n} \geq \frac \a 4.
\]
This gives a contradiction with \eqref{eq-poisson-bracket-1} and concludes the proof.
\end{proof}

\end{document}